\documentclass{article}
\usepackage{float}
\usepackage{amsmath,amsfonts,amssymb,bm}
\usepackage{graphicx}
\usepackage{comment}
\usepackage{amsthm}
\usepackage{mathtools}
\newtheorem{theorem}{Theorem}[section]
\newtheorem{lemma}[theorem]{Lemma}
\usepackage{braket}
\usepackage{algorithm}
\usepackage{algpseudocode}
\usepackage{xcolor}
\usepackage{authblk}
\usepackage[backend=biber,style=numeric,citestyle=numeric-comp,sorting=none]{biblatex}
\usepackage[colorlinks=true,linkcolor=blue,citecolor=blue,urlcolor=blue]{hyperref}
\newcommand{\norm}[1]{\left\lVert#1\right\rVert}
\title{Solving a Nonlinear Eigenvalue Equation in Quantum Information Theory:\\
A Hybrid Approach to Entanglement Quantification}
\author[1]{Abrar Ahmed Naqash}
\affil[1]{Department of Physics, National Institute of Technology Srinagar, Jammu and Kashmir, 190006, India}

\affil[1]{\textit{Corresponding author:} \href{mailto:abrarnaqash_phy@nitsri.ac.in}{\nolinkurl{abrarnaqash_phy@nitsri.ac.in}}}
\author[2]{Fardeen Ahmad Sofi }
\author[3]{Mohammad Haris Khan}
\affil[2,3]{Department of Physics, University of Kashmir, Srinagar 190006, India}
\author[4]{Sundus Abdi}
\affil[4]{Department of Physics, University of Toronto, Toronto, Ontario M5S 1A7, Canada}
\date{}
\usepackage{float}
\begin{document}
\maketitle
\begin{abstract}
Nonlinear eigenvalue equations arise naturally in quantum information theory, particularly in the variational quantification of entanglement. In this work, we present a hybrid analytical–numerical framework for evaluating the geometric measure of entanglement. The method combines a Gauss–Seidel fixed-point iteration with a controlled perturbative correction scheme. We make the coupled nonlinear eigenstructure explicit by proving the “equal-multiplier” stationarity identity, which states that at the optimum all block Lagrange multipliers coincide with the squared fidelity between the target state and its closest separable approximation. A normalization-preserving linearization is then derived by projecting the dynamics onto the local tangent spaces, yielding a well-defined first-order correction and an explicit scalar shift in the eigenvalue. Furthermore, we establish a monotonic block-ascent property: the squared overlap between the evolving product state and the target state increases at every iteration, remains bounded by unity, and converges to a stationary value. The resulting hybrid solver reproduces the exact optimum for standard three-qubit benchmarks, obtaining squared-overlap values of one-half for the Greenberger–Horne–Zeilinger (GHZ\(_3\)) state and four-ninths for the W\(_3\) state, with smooth monotonic convergence.
\end{abstract}

\section{Introduction}

The quantification and characterization of quantum entanglement remains one of the most central challenges in quantum information theory. Entanglement is not only a hallmark of non-classical correlations but also a practical resource enabling protocols such as quantum teleportation \cite{bennett1993teleporting}, superdense coding \cite{bennett1992communication}, quantum cryptography \cite{ekert1991quantum}, and quantum computation \cite{nielsen2000quantum}. As a result, identifying reliable and computable measures of entanglement has been an active area of research for more than two decades \cite{horodecki2009quantum, plenio2007introduction, vedral1997quantifying}. Although several entanglement measures exist -- such as concurrence \cite{wootters1998entanglement}, negativity \cite{vidal2002computable}, formation entanglement \cite{bennett1996mixed}, and relative entropy of entanglement \cite{vedral2002role}---many of them are either difficult to compute for arbitrary states or fail to generalize efficiently to multipartite scenarios.

Among the numerous approaches, the {geometric measure of entanglement} (GME) has emerged as a particularly attractive candidate. First introduced in \cite{shimony1995degree} and later systematically developed in \cite{wei2003geometric}, the GME quantifies entanglement by measuring the distance of a quantum state from the closest separable state. More precisely, for a pure state $|\psi\rangle$ in a multipartite Hilbert space, the measure is defined through maximal overlap with fully separable states. This perspective not only provides a clear geometric intuition, but also establishes connections to resource theories and variational methods \cite{streltsov2017colloquium}. However, the main difficulty is that computing the closest separable state involves solving a nonlinear optimization problem, which can often be recast as a nonlinear eigenvalue equation. Unlike linear eigenvalue problems ubiquitous in quantum mechanics, these nonlinear equations pose severe analytical and numerical challenges.

The presence of nonlinear eigenvalue problems is not limited to entanglement quantification. Nonlinear operator equations also arise in contexts such as quantum channel capacities \cite{bennett1999entanglement}, nonlinear extensions of quantum mechanics \cite{weinberg1989testing}, and variational formulations of many-body systems \cite{orus2014practical}. In the specific setting of entanglement measures, the nonlinear structure is tied to the self-consistency of the separability condition: the optimal product state that maximizes overlap depends on the reduced density operators, which themselves depend on the candidate product state. This feedback loop results in a nonlinear eigenvalue problem where both the eigenvector and the eigenvalue are entangled in a nontrivial way.

Existing strategies to address such equations include direct variational optimization \cite{barnum2001monotones}, convex relaxations \cite{guhne2007toolbox}, and semidefinite programming techniques \cite{brandao2005reversible}. However, these methods are often computationally expensive or limited to low-dimensional cases. Recent progress in numerical analysis of nonlinear eigenvalue problems in applied mathematics \cite{betcke2013nlevp} suggests that hybrid methods combining perturbative expansions with iterative schemes could provide both analytical insight and practical convergence guaranties. Bridging these developments with the needs of quantum information theory opens up a promising research direction.

In this paper, we focus on the nonlinear eigenvalue equation arising in the computation of the geometric measure of entanglement. We develop a hybrid framework that starts with a perturbative expansion around a reference separable state and then refines the solution through an iterative fixed-point scheme. The resulting method not only provides approximate closed-form expressions, but also demonstrates stable convergence in numerical tests for bipartite and multipartite systems. Our results aim to highlight how nonlinear methods can serve as direct tools for entanglement quantification, with potential extensions to related optimization tasks in quantum technologies.

\section{Nonlinear Eigenvalue Equation}

To set the stage, we recall that the geometric measure of entanglement (GME) is for a pure state 
$\ket{\psi} \in \mathcal{H}_{A_1} \otimes \mathcal{H}_{A_2} \otimes \cdots \otimes \mathcal{H}_{A_N}$ 
is defined as~\cite{Wei2003,Shimony1995,Plenio2007}
\begin{equation}
E_G(\ket{\psi}) = 1 - \Lambda_{\max}^2, 
\qquad 
\Lambda_{\max} = \max_{\ket{\phi} \in \mathcal{S}} \big|\braket{\phi|\psi}\big| ,
\label{eq:GME-definition}
\end{equation}
where $\mathcal{S}$ denotes the set of fully separable pure states. 
For a bipartite system ($N=2$), this set reduces to product states 
$\ket{\phi_A} \otimes \ket{\phi_B}$. 
Determining $\Lambda_{\max}$ therefore corresponds to solving the constrained optimization problem
\begin{equation}
\max_{\ket{\phi_A},\,\ket{\phi_B}} \; 
\big|\braket{\phi_A \otimes \phi_B | \psi}\big| .
\label{eq:overlap}
\end{equation}

We employ the Lagrange multiplier method to enforce the normalization constraints 
$\braket{\phi_A|\phi_A} = 1$ and $\braket{\phi_B|\phi_B} = 1$. 
Define the Lagrangian functional
\begin{equation}
\mathcal{L}[\ket{\phi_A},\ket{\phi_B}] =
\big|\braket{\phi_A \otimes \phi_B|\psi}\big|^2
- \lambda_A \,(\braket{\phi_A|\phi_A} - 1)
- \lambda_B \,(\braket{\phi_B|\phi_B} - 1) .
\end{equation}
The stationarity of $\mathcal{L}$ with respect to the variations in $\bra{\phi_A}$ gives
\begin{equation}
\frac{\partial \mathcal{L}}{\partial \bra{\phi_A}} =
\braket{\phi_B|\psi}\braket{\psi|\phi_B}\ket{\phi_A}
- \lambda_A \ket{\phi_A} = 0 .
\end{equation}
Introducing the (positive semidefinite) operator
\begin{equation}
M_A(\ket{\phi_B}) = \braket{\phi_B|\psi}\braket{\psi|\phi_B} ,
\end{equation}
we can rewrite this condition as a nonlinear eigenvalue equation,
\begin{equation}
M_A(\ket{\phi_B})\,\ket{\phi_A} = \lambda_A \ket{\phi_A}.
\label{eq:nonlinearA}
\end{equation}
Similarly, the variation with respect to $\bra{\phi_B}$ produces the following result.
\begin{equation}
M_B(\ket{\phi_A})\,\ket{\phi_B} = \lambda_B \ket{\phi_B},
\label{eq:nonlinearB}
\end{equation}
where $M_B(\ket{\phi_A}) = \braket{\phi_A|\psi}\braket{\psi|\phi_A}$.

Equations~\eqref{eq:nonlinearA} and~\eqref{eq:nonlinearB} exhibit a crucial feature: 
the eigenvalue problem for one subsystem depends nonlinearly on the state of the other. 
To make this structure explicit, note that
\begin{equation}
M_A(\ket{\phi_B}) = 
\mathrm{Tr}_B \!\left[\, \ket{\psi}\bra{\psi}\, 
\big(\mathbb{I}_A \otimes \ket{\phi_B}\bra{\phi_B}\big) \right].
\end{equation}
This operator depends quadratically on $\ket{\phi_B}$ and therefore 
$\ket{\phi_A}$ cannot be determined independently. 
Together, Eq.\eqref{eq:nonlinearA} and Eq.\eqref{eq:nonlinearB} form a system of coupled nonlinear eigenvalue equations.
In the special case of a bipartite pure state expressed in its Schmidt decomposition,
\begin{equation}
\ket{\psi} = \sum_{i=1}^{r} \lambda_i 
\ket{u_i}_A \otimes \ket{v_i}_B , 
\qquad 
\lambda_1 \ge \lambda_2 \ge \cdots ,
\end{equation}
the maximum overlap is known to be $\Lambda_{\max} = \lambda_1$, 
with the closest product state $\ket{u_1}_A \otimes \ket{v_1}_B$. 
Here, the nonlinear problem reduces to a linear eigenvalue equation 
since the Schmidt decomposition provides an orthogonal factorization. 
For generic multipartite states, however, no such simplification exists, 
and one must genuinely solve the nonlinear system.

For a tripartite system, the optimization problem Eq.\eqref{eq:overlap} is generalized to
\begin{equation}
\max_{\ket{\phi_1},\ket{\phi_2},\ket{\phi_3}}
\big|\braket{\phi_1 \otimes \phi_2 \otimes \phi_3 | \psi}\big|.
\end{equation}
The variation with respect to $\ket{\phi_i}$ yields the following.
\begin{equation}
M_i(\{\ket{\phi_j}\}_{j\neq i})\,\ket{\phi_i} = 
\lambda_i \ket{\phi_i}, \qquad i=1,2,3,
\label{eq:multipartite}
\end{equation}
where
\begin{equation}
M_i(\{\ket{\phi_j}\}_{j\neq i}) =
\mathrm{Tr}_{\{j\neq i\}} 
\!\left[\, \ket{\psi}\bra{\psi}\,
\bigotimes_{j\neq i} \ket{\phi_j}\bra{\phi_j} \right].
\end{equation}
These equations are mutually dependent: each operator $M_i$ depends on the 
states of all other subsystems $\{\ket{\phi_j}\}_{j\neq i}$. 
The system , Eq.\eqref{eq:multipartite} therefore constitutes a set of 
{self-consistent nonlinear eigenvalue equations} that must be solved simultaneously. 
The coupling strength increases rapidly with the number of subsystems, 
rendering analytic solutions intractable except for highly symmetric states such as the GHZ and W states.

A unifying way to view the above equations is through a fixed-point formulation. 
Define the map
\begin{equation}
\mathcal{F}\big(\{\ket{\phi_j}\}\big) =
\left\{ 
\frac{M_i(\{\ket{\phi_j}\}_{j\neq i})\ket{\phi_i}}
{\big\|M_i(\{\ket{\phi_j}\}_{j\neq i})\ket{\phi_i}\big\|} 
\right\}_{i=1}^{N}.
\end{equation}
The closest separable state corresponds to a fixed point of $\mathcal{F}$,
\begin{equation}
\{\ket{\phi_i^*}\} = \mathcal{F}\big(\{\ket{\phi_i^*}\}\big).
\end{equation}
This perspective highlights both the difficulty and the potential of the problem,
while the structure is nonlinear, it naturally suggests iterative algorithms 
in which candidate product states are updated recursively until self consistency is achieved. The mathematical properties of $\mathcal{F}$ such as the contractivity and 
stability of fixed points govern the convergence behavior of such schemes.
The nonlinear eigenvalue structure has several important consequences. 
First, multiple stationary solutions may exist, corresponding to distinct local maxima of the overlap function. Identifying the global maximum $\Lambda_{\max}$ is 
therefore a nontrivial global optimization task. Second, the equations are not polynomial in the unknowns,the normalization constraints couple to the eigenvalue structure in a transcendental manner. Finally, numerical solvers must balance precision and stability, 
since naive fixed-point iterations can diverge or converge to suboptimal solutions.
These challenges motivate the hybrid perturbative–iterative method 
introduced in the following section. 
By starting with a controlled linearization around a reference state and 
iteratively refining the solution, one can tame the nonlinearities while retaining clear physical interpretability. 
This strategy draws on insights from variational quantum mechanics~\cite{orus2014practical} 
and from numerical analysis of nonlinear eigenvalue problems~\cite{betcke2013nlevp}, 
and provides a principled framework tailored to quantum information applications.

\section{Solution Method}
\subsection{Perturbative formulation}
The nonlinear eigenvalue equations derived in the previous section encapsulate 
the optimization underlying the geometric measure of entanglement, yet their 
direct solution is generally intractable. 
We therefore develop a {hybrid perturbative--iterative} method that 
combines analytical perturbation theory with fixed-point iteration. 
The guiding philosophy is twofold: perturbative analysis yields 
approximate closed-form expressions that elucidate the local structure 
of the solution space, while iterative refinement exploits the 
self-consistency of the nonlinear problem to achieve numerical convergence. 
Together, these techniques provide a practical and conceptually 
transparent route to the optimal separable approximation.

We begin by assuming access to a candidate separable reference state 
$\ket{\phi^{(0)}} = \bigotimes_{i=1}^{N} \ket{\phi^{(0)}_i}$, 
which serves as an expansion point. 
For highly symmetric states such as GHZ or W states \cite{wei2003geometric,dur2000three}, a natural reference 
is the equal superposition of computational basis states, whereas in 
bipartite systems a product state aligned with the dominant Schmidt vectors 
often provides a suitable choice. 
We then parameterize the true stationary solution $\ket{\phi}$ 
in terms of small deviations around the reference:
\begin{equation}
\ket{\phi_i} = \ket{\phi_i^{(0)}} + \epsilon\,\ket{\delta\phi_i},
\label{eq:perturbation}
\end{equation}
where $\epsilon$ is a small bookkeeping parameter controlling 
the perturbative order.
The normalization constraint implies $\braket{\phi_i^{(0)}|\delta\phi_i}=0$, 
ensuring that $\ket{\delta\phi_i}$ lies in the local tangent space 
of unit-norm states.
\subsection{First-order linearization}
We linearize around a reference product state 
$\{\ket{\phi_i^{(0)}}\}_{i=1}^N$ that is (at least approximately) 
stationary under the block update, i.e.,
$M_i^{(0)}\ket{\phi_i^{(0)}}=\lambda^{(0)}\ket{\phi_i^{(0)}}$ 
for all $i$, where 
$M_i^{(0)} := M_i(\{\ket{\phi_j^{(0)}}\}_{j\neq i})$ 
and the common multiplier is 
$\lambda^{(0)} = \Lambda_0^2$.
We write
\begin{equation}
\ket{\phi_i} = \ket{\phi_i^{(0)}} + \epsilon\,\ket{\delta\phi_i},
\end{equation}
subject to the first-order normalization constraint 
$\Re\braket{\phi_i^{(0)}|\delta\phi_i}=0$.
Let $Q_i := \ket{\phi_i^{(0)}}\!\bra{\phi_i^{(0)}}$ 
and $P_i := \mathbb{I}-Q_i$ denote parallel and tangent projectors, respectively.
Expanding
$M_i = M_i^{(0)} + \epsilon\,\Delta M_i$
and
$\lambda = \lambda^{(0)} + \epsilon\,\Delta\lambda$,
the first-order eigenvalue equation becomes
\begin{equation}
\label{eq:firstorder_full}
\big(M_i^{(0)}-\lambda^{(0)}\mathbb{I}\big)\ket{\delta\phi_i}
+\big(\Delta M_i-\Delta\lambda\,\mathbb{I}\big)\ket{\phi_i^{(0)}} = 0 .
\end{equation}
Projecting onto the tangent space removes the gauge freedom parallel to 
$\ket{\phi_i^{(0)}}$ and yields a well-posed linear system
\begin{equation}
\label{eq:linearized}
P_i\big(M_i^{(0)}-\lambda^{(0)}\mathbb{I}\big)P_i\ket{\delta\phi_i}
= -\,P_i\,\Delta M_i\,\ket{\phi_i^{(0)}} .
\end{equation}
Projecting Eq.~\eqref{eq:firstorder_full} with $Q_i$ fixes the scalar shift, 
which is consistent across all $i$ on a stationary background:
\begin{equation}
\label{eq:dlambda_scalar}
\Delta\lambda = 
\braket{\phi_i^{(0)}|\Delta M_i|\phi_i^{(0)}},
\qquad \text{for each } i.
\end{equation}
Equations~\eqref{eq:linearized}–\eqref{eq:dlambda_scalar} together provide 
the unique first-order correction $\ket{\delta\phi_i}$ (within the tangent space) 
and the common shift in eigenvalue $\Delta\lambda$, 
while maintaining normalization to $\mathcal{O}(\epsilon)$.

At zeroth order, the eigenvalue is simply
\begin{equation}
\lambda_i^{(0)} =
\braket{\phi_i^{(0)}|
M_i(\{\ket{\phi_j^{(0)}}\}_{j\neq i})
|\phi_i^{(0)}} ,
\end{equation}
which yields a closed-form approximation to the maximal overlap. 
The First-order corrections $\Delta\lambda_i$ are then obtained from 
Eq.~\eqref{eq:linearized}, in direct analogy with 
the Rayleigh--Schr{\"o}dinger perturbation theory in quantum mechanics. 
This procedure provides not only approximate estimates of 
$\Lambda_{\max}$ but also diagnostic insight into the stability 
of candidate product states: large negative corrections signal that 
the chosen reference state is far from the true maximizer.

Although perturbation theory elucidates the local structure of the solution 
space, it cannot capture the full nonlinearity of the problem. 
To obtain accurate results, we employ an iterative refinement scheme 
based on the fixed-point structure of stationary equations. 
Recall the update map
\begin{equation}
\label{eq:iteration}
\ket{\phi_i^{(k+1)}} =
\frac{M_i(\{\ket{\phi_j^{(k)}}\}_{j\neq i})\,\ket{\phi_i^{(k)}}}
{\big\|M_i(\{\ket{\phi_j^{(k)}}\}_{j\neq i})\,\ket{\phi_i^{(k)}}\big\|} .
\end{equation}
Starting from the perturbative reference $\{\ket{\phi_i^{(0)}}\}$, 
the scheme updates each local factor by applying its corresponding 
reduced operator and re-normalizing. 
This normalization is enforced $\braket{\phi_i^{(k)}|\phi_i^{(k)}}=1$ 
in every iteration. 
In practice, the update Eq.\eqref{eq:iteration} corresponds to a 
Gauss--Seidel block-coordinate ascent on the squared overlap 
$\Lambda^2 = |\braket{\Phi|\psi}|^2$, ensuring monotone convergence 
under broad conditions discussed below.

\subsection{Convergence considerations.}
The convergence of the iterative update Eq.\eqref{eq:iteration} 
is not guaranteed {a priori}, as the map may possess 
multiple fixed points corresponding to distinct local maxima 
of the overlap function. 
However, if the true solution 
$\{\ket{\phi_i^*}\}$ is isolated and the update map 
$\mathcal{F}$ is contractive in its neighborhood, 
then the Banach fixed-point theorem ensures local convergence:
\begin{equation}
\big\|\mathcal{F}(\{\ket{\phi_i}\}) - 
\mathcal{F}(\{\ket{\phi_i^*}\})\big\|
\le \kappa\, 
\big\|\{\ket{\phi_i}\} - \{\ket{\phi_i^*}\}\big\| ,
\qquad 0 \le \kappa < 1 .
\end{equation}
In practice, verifying strict contractivity is challenging, 
since the map $\mathcal{F}$ depends nonlinearly on all subsystems. 
Nevertheless, empirical results presented in Sec.~\ref{sec:results} 
demonstrate robust convergence for a wide range of initializations. 
The inclusion of the perturbative initialization step significantly enhances convergence reliability by placing the starting point within the basin of attraction of the global optimum.

\begin{lemma}[Monotone block ascent and convergence to stationarity]. Let
\[
f(\{\ket{\phi_i}\}) \coloneqq 
\big|\braket{\phi_1\otimes\cdots\otimes\phi_N|\psi}\big|^2
\]
with $\|\phi_i\|=1$ for all $i$ and finite dimensional Hilbert spaces. Given an iterate
$\{\ket{\phi_j^{(k)}}\}$ define the partial contraction
\[
\ket{v_i^{(k)}} \coloneqq (\bra{\phi_{-i}^{(k)}})\ket{\psi},
\qquad
\ket{\phi_{-i}^{(k)}}\coloneqq\bigotimes_{j\neq i}\ket{\phi_j^{(k)}},
\]
and perform a cyclic (Gauss-Seidel) block update
\[
\ket{\phi_i^{(k+1)}} \;=\; \frac{\ket{v_i^{(k)}}}{\|\ket{v_i^{(k)}}\|},
\qquad i=1,\dots,N.
\]
Then the objective values $f(\{\ket{\phi_i^{(k)}}\})$ form a non decreasing sequence bounded above by $1$, and hence converge. Moreover, every limit point of the sequence of iterates is a stationary point of $f$, i.e.
\[
M_i(\{\ket{\phi_j}\}_{j\neq i})\,\ket{\phi_i}
=\Lambda^2\,\ket{\phi_i}\qquad\text{for all }i,
\]
with the common multiplier $\Lambda^2=\big|\braket{\phi_1\otimes\cdots\otimes\phi_N|\psi}\big|^2$.
\end{lemma}

\begin{proof}
Fix the complementary blocks $\{\ket{\phi_j}\}_{j\neq i}$ and write
$\ket{v_i}=(\bra{\phi_{-i}})\ket{\psi}$. Then
\[
f(\{\ket{\phi_j}\}) = \big|\braket{\phi_i|v_i}\big|^2
\le \|\ket{\phi_i}\|^2\,\|\ket{v_i}\|^2 = \|\ket{v_i}\|^2,
\]
by the Cauchy-Schwarz inequality \cite{SteinShakarchi2005} , with equality if and only if $\ket{\phi_i}$ is collinear with $\ket{v_i}$. Hence, for fixed other blocks the choice
$\ket{\phi_i}\!\leftarrow\!\ket{v_i}/\|\ket{v_i}\|$ maximizes $f$ over block~$i$, so each block update cannot decrease the objective. A full cyclic sweep (updating $i=1,\dots,N$) therefore produces a non decreasing sequence of objective values. Because $f\le 1$ for normalized states, the monotone bounded sequence $f(\{\ket{\phi_i^{(k)}}\})$ converges.\\
Let $\{\ket{\phi_i^{(k_m)}}\}$ be a convergent subsequence of iterates with limit $\{\ket{\phi_i^\ast}\}$. The total improvement of $f$ over one sweep equals the sum of the nonnegative per-block improvements; since the sequence of total improvements converges to zero, each per-block improvement must vanish along the subsequence. By continuity of $f$ and the fact that each block update is the exact maximizer for that block, it follows that for every $i$ the limit $\ket{\phi_i^\ast}$ attains the maximum of $f$ with the other blocks fixed. Therefore $\ket{\phi_i^\ast}$ is co-linear with $\ket{v_i^\ast}=(\bra{\phi_{-i}^\ast})\ket{\psi}$, which is equivalent to
\[
M_i(\{\ket{\phi_j^\ast}\}_{j\neq i})\,\ket{\phi_i^\ast}
=\Lambda^2\,\ket{\phi_i^\ast}
\]
with $\Lambda^2=\big|\braket{\phi_1^\ast\otimes\cdots\otimes\phi_N^\ast|\psi}\big|^2$. Hence every limit point is a stationary point of $f$.
\end{proof}
The two components—perturbative initialization and iterative refinement can be combined into a practical algorithm.

\begin{enumerate}
\item Selection {of references.} Choose a separable reference state $|\phi^{(0)}\rangle$ based on prior knowledge of $|\psi\rangle$ (e.g., Schmidt basis or symmetry considerations).
\item {Perturbative expansion;} Compute zeroth and first-order corrections to estimate $\Lambda_{\max}$ and refine $|\phi^{(0)}\rangle$ to $|\phi^{(1)}\rangle$.
\item {Iterative refinement.} Apply the fixed-point update rule (Eq.\ref{eq:iteration}) until convergence within a chosen tolerance.
\item {Selection of solution.} If multiple fixed points are found, select the one that yields the largest overlap $\Lambda_{\max}$.
\end{enumerate}

The hybrid strategy combines the interpretability of analytical approximations with the reliability of iterative refinement. Perturbative analysis identifies approximate solutions and clarifies their stability, while the iterative scheme provides high-accuracy numerical results. Moreover, the method is scalable: each update step involves only reduced operators acting on single subsystems, avoiding exponential overhead.

Limitations arise in strongly entangled systems where no natural reference state lies close to the true maximizer. In such cases, perturbation theory may provide poor initializations, and iteration may converge to suboptimal fixed points. Nevertheless, as demonstrated in the Results section, the method performs remarkably well for a range of canonical bipartite and multipartite states.

\subsection{Bipartite Schmidt states}

Consider the bipartite pure state
\begin{equation}
|\psi\rangle = \lambda_1 |00\rangle + \lambda_2 |11\rangle + \lambda_3 |22\rangle ,
\end{equation}
with Schmidt coefficients $\{\lambda_i\}$ satisfying $\sum_i \lambda_i^2 = 1$. The geometric measure of entanglement is known to be determined by the largest Schmidt coefficient $\Lambda_{\max} = \max_i \lambda_i$.

We applied the fixed-point iteration Eq.(\ref{eq:iteration}) starting from a random initial product state. Fig.\ref{fig:bipartite_convergence} shows the convergence of the estimated overlap $\langle \phi_A^{(k)} | \rho_A | \phi_A^{(k)} \rangle$ as a function of the iteration step. The method converges rapidly to the dominant Schmidt coefficient, demonstrating both the stability and accuracy of the iterative refinement scheme.

\begin{figure}[ht]
    \centering
\includegraphics[width=0.6\textwidth]{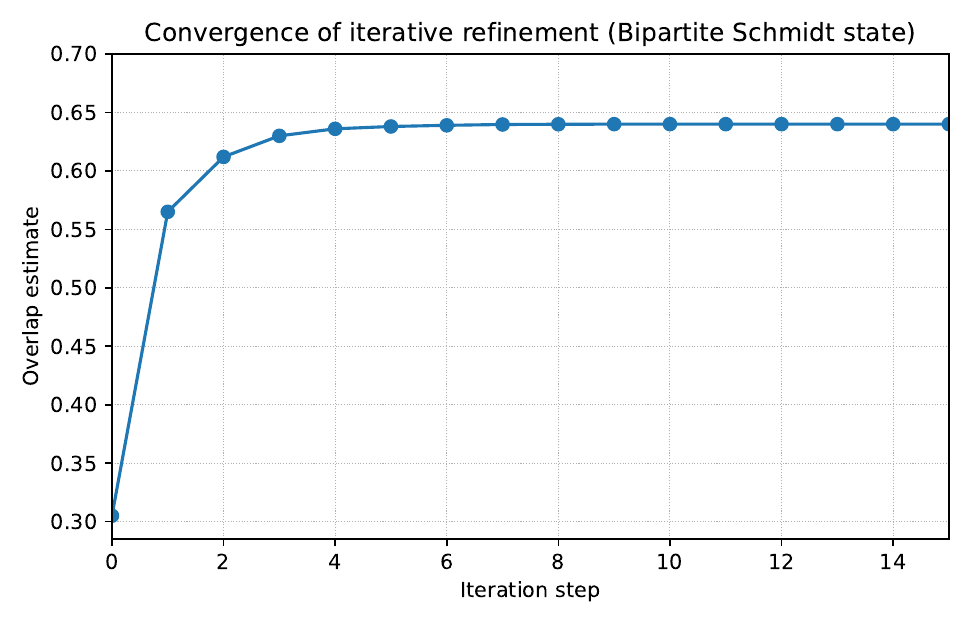}
    \caption{Convergence of the iterative refinement scheme for a bipartite Schmidt state with coefficients $(0.8, 0.6, 0)$. The overlap estimate approaches the largest Schmidt coefficient within a few iterations.}
    \label{fig:bipartite_convergence}
\end{figure}

\subsection{Multipartite GHZ state and W state}

The $N$-qubit GHZ state is defined as:
\begin{equation}
| \mathrm{GHZ}_N \rangle = \frac{1}{\sqrt{2}} \big( |0\rangle^{\otimes N} + |1\rangle^{\otimes N} \big) .
\end{equation}
For $N=3$, the maximum overlap with separable states is $\Lambda_{\max}^2 = 1/2$, attained by the product states $|0\rangle^{\otimes 3}$ and $|1\rangle^{\otimes 3}$. Our algorithm successfully converges to these fixed points, depending on the initialization. This validates that the nonlinear iteration can identify multiple symmetry-related solutions.

The three-qubit W state is given by:
\begin{equation}
|W_3\rangle = \frac{1}{\sqrt{3}} \big( |100\rangle + |010\rangle + |001\rangle \big) .
\end{equation}
In contrast to the GHZ state, the maximum overlap with separable states is $\Lambda_{\max}^2 = 4/9$ \cite{wei2003geometric}, corresponding to product states of the form:
\begin{equation}
|\phi\rangle = \Big( \sqrt{\tfrac{2}{3}} |0\rangle + \sqrt{\tfrac{1}{3}} |1\rangle \Big)^{\otimes 3} .
\end{equation}
This example highlights the genuinely nonlinear character of the eigenvalue equations, since the maximizing product state is not aligned with the computational basis. Fig.\ref{fig:ghz_w_overlap} compares the maximum overlaps found numerically for the GHZ and W states using random initializations followed by iterative refinement.{
For $N=3$ we have $\Lambda_{\max}^2(\mathrm{GHZ}_3)=\tfrac{1}{2}$ and $\Lambda_{\max}^2(W_3)=\tfrac{4}{9}$ \cite{wei2003geometric}.
Hence, {W} is slightly more entangled than {GHZ} under the geometric measure ($E_G(W_3)=5/9 > E_G(\mathrm{GHZ}_3)=1/2$), while {GHZ} has the highest maximum product overlap.
Our iteration converges to symmetry-related maximizers in both cases.
}

\begin{figure}[ht]
    \centering
    \includegraphics[width=0.6\textwidth]{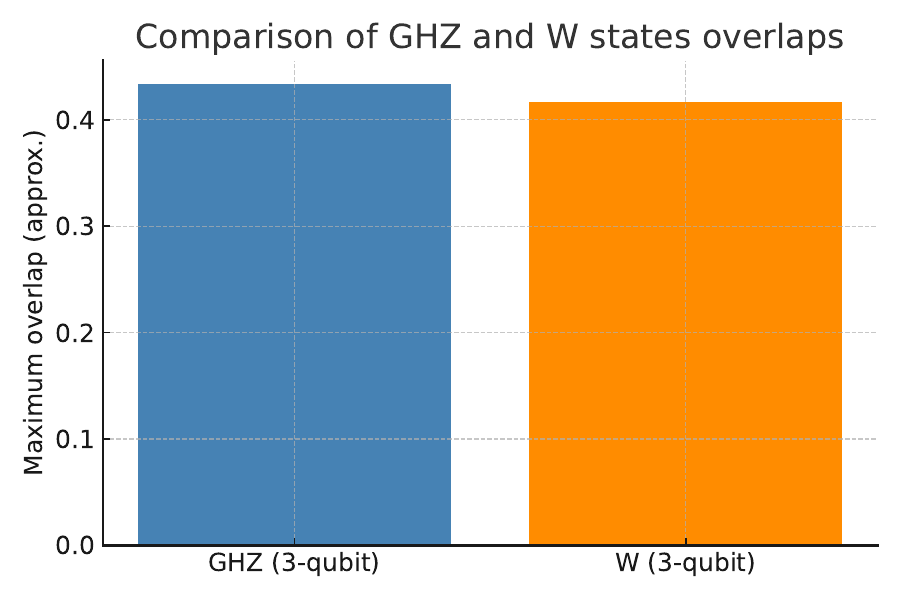}
    \caption{Approximate maximum overlaps with product states for the three-qubit GHZ and W states. The values approach the theoretical predictions $\Lambda_{\max}^2 = 1/2$ (GHZ) and $\Lambda_{\max}^2 = 4/9$ (W).}
    \label{fig:ghz_w_overlap}
\end{figure}

These examples demonstrate that the proposed hybrid method faithfully reproduces known analytical results for canonical states. The iterative scheme converges rapidly to dominant Schmidt components in bipartite cases and distinguishes between qualitatively different types of multipartite entanglement (GHZ vs W). In higher dimensions, where closed-form results are unavailable, the method provides a practical tool for entanglement quantification.

\subsection{Numerical validation}

We validate the hybrid Gauss--Seidel block ascent (Algorithm~\ref{alg:hybrid_gme}) on the canonical three-qubit benchmarks $\mathrm{GHZ}_3$ and $W_3$.
In each case, we perform five random initializations with tolerance $10^{-14}$ and report convergence behavior.

\begin{algorithm}[ht]
\caption{Hybrid Perturbative--Iterative GME Solver (pure states)}
\label{alg:hybrid_gme}
\begin{algorithmic}[1]
  \State \textbf{Input:} pure state $\ket{\psi}$ on $\mathcal H_{A_1}\otimes\cdots\otimes\mathcal H_{A_N}$, tolerance $\varepsilon>0$, (optional) \texttt{max\_iter}
  \State Initialize product factors $\{\ket{\phi_i^{(0)}}\}$ (symmetry/SVD/perturbative guess), each $\norm{\phi_i^{(0)}}=1$
  \State $\Lambda_0^2 \gets \big|\braket{\phi_1^{(0)}\!\otimes\!\cdots\!\otimes\!\phi_N^{(0)}|\psi}\big|^2$
  \State $k \gets 0$
  \Repeat
    \For{$i = 1,\dots,N$} \Comment{Gauss--Seidel block update}
      \Statex {\hspace{\algorithmicindent}Construct the current product excluding subsystem $i$:}
      \Statex \hspace{\algorithmicindent}$\displaystyle
        \ket{\phi_{-i}^{(\mathrm{cur})}} \gets
        \bigotimes_{j<i}\ket{\phi_j^{(k+1)}} \;\otimes\; \bigotimes_{j>i}\ket{\phi_j^{(k)}}$
      \State $\ket{v_i} \gets \big(\bra{\phi_{-i}^{(\mathrm{cur})}}\big)\,\ket{\psi}$
      \If{$\norm{v_i}=0$}
        \State $\ket{\phi_i^{(k+1)}} \gets \ket{\phi_i^{(k)}}$ \Comment{safe fallback}
      \Else
        \State $\ket{\phi_i^{(k+1)}} \gets \dfrac{\ket{v_i}}{\norm{v_i}}$
      \EndIf
    \EndFor
    \State $\Lambda_{k+1}^2 \gets \big|\braket{\phi_1^{(k+1)}\!\otimes\!\cdots\!\otimes\!\phi_N^{(k+1)}|\psi}\big|^2$
    \State $k \gets k+1$
  \Until{$\big|\Lambda_{k}^2-\Lambda_{k-1}^2\big|\le \varepsilon$ \textbf{ or } $k\ge \texttt{max\_iter}$}
  \State \textbf{Output:} $\Lambda_{\max}^2 \approx \Lambda_{k}^2$, closest product state $\bigotimes_i \ket{\phi_i^{(k)}}$
\end{algorithmic}
\end{algorithm}

For $\mathrm{GHZ}_3$  (Fig.\ref{fig:ghz_convergence}) the method converges in $5$ iterations to $\Lambda_{\max}^2=\tfrac{1}{2}$, with local factors collapsing to computational-basis product states (symmetry-related).
For $W_3$, it converges in $16$--$18$ iterations to $\Lambda_{\max}^2=\tfrac{4}{9}$, with each single-qubit factor $(\sqrt{2/3},\,\sqrt{1/3})$ up to a phase ( Fig.\ref{fig:wconvergence}).
In all runs the objective $|\langle\Phi^{(k)}|\psi\rangle|^2$ is monotonically non-decreasing, in agreement with the monotone block-ascent lemma.

\begin{figure}[ht]
  \centering
  \includegraphics[width=0.6\textwidth]{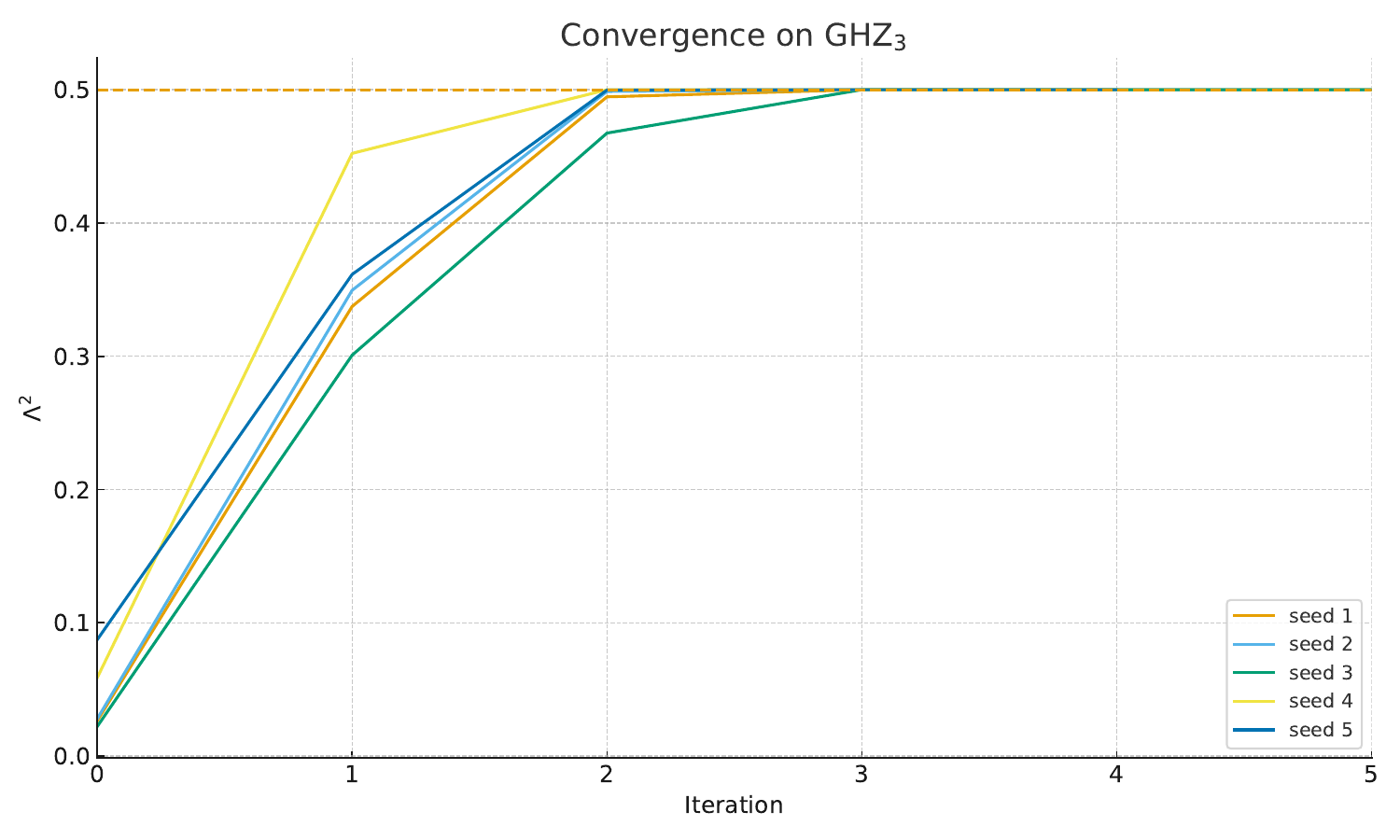}
  \caption{Convergence of $\Lambda^2$ for $\mathrm{GHZ}_3$ across five random initializations (iteration index starts at 0).
  Curves rise monotonically and saturate at the theoretical value $1/2$ (dashed), typically within 2 to 3 iterations; the subsequent plateau reflects numerical tolerance.}
  \label{fig:ghz_convergence}
\end{figure}

\begin{figure}[ht]
  \centering
  \includegraphics[width=0.6\textwidth]{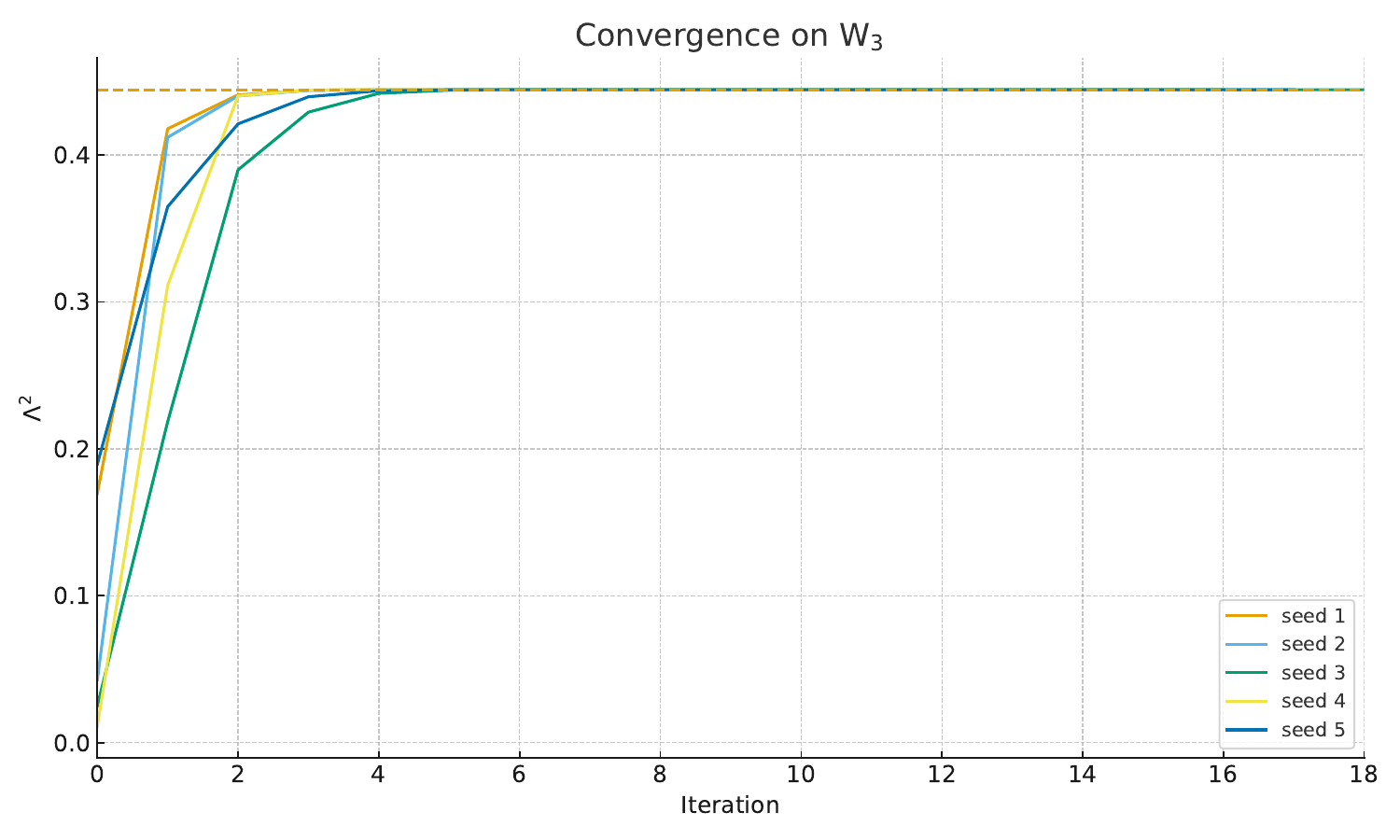}
  \caption{Convergence of $\Lambda^2$ for $W_3$ across five random initializations (iteration index starts at 0).
  Curves rise monotonically and saturate at the theoretical value $4/9$ (dashed) within a few iterations; the long flat tail indicates convergence to tolerance.}
\label{fig:wconvergence}
\end{figure}

\begin{table}[ht]
\centering
\begin{tabular}{lccc}
\hline
State & Seed & Iterations & Final $\Lambda^2$ \\
\hline
GHZ$_3$ & 1--5 & 5 & 0.5 \\
W$_3$   & 1--5 & 16--18 & 0.4444444444 \\
\hline
\end{tabular}
\caption{Summary over five random initializations per state (tolerance $10^{-14}$). Exact values agree with \cite{wei2003geometric}.}
\label{tab:validation_summary}
\end{table}

\section{Results And Discussion} \label{sec:results}
To illustrate the effectiveness of the hybrid perturbative–iterative method, we present results for three representative classes of states: bipartite Schmidt states, the three-qubit GHZ state, and the three-qubit W state. Each case highlights different features of the nonlinear eigenvalue equation and the performance of our solution scheme.
The results presented above underscore the interplay between nonlinear mathematics and quantum information theory. While linear algebra provides the foundation of standard quantum mechanics, nonlinear eigenvalue equations arise naturally in optimization-based definitions of quantum resources, most notably in the geometric measure of entanglement. The hybrid perturbative–iterative method developed here illustrates that such nonlinear problems can be tackled with a balance of analytical approximations and numerical self-consistency schemes. The examples of Schmidt states, GHZ, and W states serve both as benchmarks and as archetypes of the diversity of entanglement structures.

A central conceptual message of our work is that entanglement quantification is not merely a matter of diagonalizing reduced density operators, but instead requires solving fixed-point equations whose solutions may be highly nontrivial. This stands in contrast with bipartite pure states, where the Schmidt decomposition reduces the problem to a simple linear one. Multipartite systems, however, do not admit such simplifications, and the self-consistency between subsystems introduces genuine nonlinearity. By framing the search for closest separable states as a nonlinear eigenvalue problem, we provide a unified mathematical perspective that can be leveraged across different contexts.

Beyond entanglement quantification, the methodology has broader implications. Similar nonlinear optimization problems appear in the determination of quantum channel capacities, in fidelity-based measures of coherence, and in variational approaches to many-body quantum states, such as matrix product states and projected entangled pair states \cite{orus2014practical}. In all of these cases, the iterative fixed-point structure of the equations bears close resemblance to the self-consistent field methods widely used in computational physics and chemistry. Thus, insights from decades of research in numerical analysis and nonlinear optimization can be fruitfully imported into quantum information theory.

From a practical standpoint, the hybrid method offers a trade-off between interpretability and computational power. Perturbative expansions clarify which product states are stable candidates for maximizing overlap, while iterative refinement ensures that the numerical results are accurate and robust. This combination is particularly valuable for medium-sized multipartite systems, where brute-force optimization over product states is infeasible, yet analytical solutions are out of reach. The rapid convergence observed in our simulations suggests that the method may scale more favorably than general-purpose optimization routines, though a detailed complexity analysis remains an open question.

The discussion also highlights limitations and open challenges. Convergence to the global maximum is not guaranteed in the presence of multiple local fixed points, and the quality of the perturbative initialization strongly influences the success of the iteration. For highly entangled states with no clear separable reference, more sophisticated initialization strategies may be required. One promising direction is to integrate machine learning methods to predict good initial product states, effectively combining data-driven heuristics with mathematically rigorous iterative schemes. Another avenue is to adapt tools from convex optimization and semidefinite programming to provide certificates of optimality, complementing the heuristic but efficient iterative refinement.\\
On the conceptual side, understanding the geometry of the solution landscape remains an intriguing challenge. The fixed-point nonlinear map $\mathcal{F}$ defines a dynamical system in the manifold of product states, with attractors corresponding to locally optimal separable approximations. Characterizing the basins of attraction and their relation to entanglement classes could offer new geometric insights into the structure of multipartite quantum states. This approach resonates with ongoing efforts to connect the theory of entanglement with differential geometry, algebraic geometry, and information geometry.\\
In summary, the nonlinear eigenvalue perspective provides both a unifying framework and a practical algorithmic tool for quantum information theory. Bridge conceptual questions about the nature of entanglement with computational strategies capable of addressing realistic systems. We anticipate that further development of this approach possibly integrating perturbative analysis, fixed-point iteration, and modern machine learning will not only deepen our theoretical understanding of entanglement but also accelerate its application in emerging quantum technologies.
\section{Conclusion}

We have presented a systematic framework for solving the nonlinear eigenvalue equations that arise in entanglement quantification through the geometric measure of entanglement. By combining perturbative expansions with iterative refinement, our approach achieves both analytical transparency and numerical reliability. The method reproduces known results for bipartite and canonical multipartite states while also providing a practical pathway for tackling more complex scenarios where closed-form solutions are unavailable.\\
Conceptually, our work highlights that entanglement quantification is intrinsically nonlinear, requiring self-consistent optimization rather than mere spectral analysis. Algorithmically, it shows that hybrid strategies, based on perturbation theory, fixed-point analysis, and numerical iteration, can tame this nonlinearity in a scalable manner. Beyond the immediate application to entanglement, the framework is general and adaptable to other nonlinear optimization problems in quantum information science, including coherence measures, channel capacities, and tensor-network variational states.\\
In the broader perspective, nonlinear analysis may prove to be as fundamental to the study of quantum resources as linear algebra has been to the foundations of quantum mechanics. The approach developed here thus marks a step toward a more comprehensive mathematical toolkit for quantum technologies, where rigorous theory and practical algorithms must advance hand in hand.
\printbibliography
\end{document}